\documentclass[letterpaper, 10 pt, conference]{ieeeconf}  

\IEEEoverridecommandlockouts                              
\title{\LARGE \bf
Covariance Intersection-based Invariant Kalman Filtering for Distributed Pose Estimation 
}
\author{Haoying Li$^{\rm a}$, Xinghan Li$^{\rm b}$, Shuaiting Huang$^{\rm b}$, Chao Yang$^{\rm c}$, and Junfeng Wu$^{\rm a}$
\thanks{  
$^{\rm a}$:~School of Data Science,  The Chinese University of Hong Kong, Shenzhen, Shenzhen, P. R. China.
 $^{\rm b}$:~College of Control Science and Engineering, Zhejiang University, Hangzhou, P. R. China.
 $^{\rm c}$: Department of Automation, East China University of Science and Technology, Shanghai, P. R. China.
 Emails: haoyingli@link.cuhk.edu.cn(H. Li), xinghanli@zju.edu.cn(X. Li),  
shuait\_huang@zju.edu.cn (S. Huang), 
 yangchao@ecust.edu.cn (C. Yang), junfengwu@cuhk.edu.cn (J. Wu).
} 			
}    

\begin{document}

\maketitle
\thispagestyle{empty}
\pagestyle{empty}

\begin{abstract}
This paper presents a novel approach to distributed pose estimation in the multi-agent system based on an invariant Kalman filter with covariance intersection. Our method models uncertainties using Lie algebra and applies object-level observations within Lie groups, which have practical application value. We integrate covariance intersection to handle correlated estimates and use the invariant Kalman filter to merge independent data sources. This strategy allows us to effectively tackle the complex correlations of cooperative localization among agents, ensuring our estimates are neither too conservative nor overly confident. Additionally, we examine the consistency and stability of our algorithm, providing evidence of its reliability and effectiveness in managing multi-agent systems.
\end{abstract}

\section{INTRODUCTION}
Pose estimation has been applied in various scenarios, including unmanned driving, virtual reality, and conveyor line monitoring~\cite{dorri2018multi}.  Among these, distributed pose estimation serves as the foundation for high-level tasks in multi-agent systems, making it a valuable area of research.

For the pose estimation problem, filtering methods, especially the Kalman filter (KF) and its variants, exhibit superior real-time performance compared to optimization-based methods. Although the extended Kalman filter (EKF) has been employed~\cite{trawny2005indirect}, the nonlinearity of the state-space introduces inaccuracies. The invariant extended Kalman filter (InEKF), which is grounded in invariant observer theory~\cite{barrau2016invariant}, has been utilized in pose estimation with promising results. It evaluates uncertainty in Lie algebra and then maps it back to the Lie group, ensuring observability consistency between the log-linear system and the original nonlinear system.

Multi-agent pose estimation introduces complex correlations via neighbor information sharing, with environmental measurements uncorrelated to motion noise. However, the KF only handles known correlations, leading to challenges in the design of the Distributed Kalman Filter~(DKF). Several attempts have been made. The consensus DKF~\cite{carli2008distributed} converges to a centralized KF but requires infinite communication steps. Without explicit correlation, Covariance Intersection~(CI)~\cite{chen2002estimation} can compute consistent estimates. However, CI assumes correlation in the worst case among estimates, potentially leading to over-conservative estimates if independence exists. To address this issue, Split Covariance Intersection~(SCI)~\cite{li2013cooperative} treats dependent and independent parts separately. However, SCI is applicable when the estimation error can be decomposed into two mutually independent components, such as simultaneously estimating self-state and target state~\cite{7954723}.

Combining CI and KF yields a consistent estimate, neither too conservative nor too confident. Chang \emph{et al.}~\cite{chang2021resilient} utilizes the DKF based on CI, termed CIKF for consistent multi-robot localization. Zhu \emph{et al.}~\cite{zhu2020fully} introduces a CI-based distributed information filter for localization and target tracking. However, incorporation of Lie group modeling is absent in the approaches discussed previously. Xu \emph{et al.}~\cite{10232373} introduces a CI-based distributed InEKF on Lie groups for cooperative target tracking, yet neglecting relative noises and assuming that agents are static. Lee \emph{et al.}~\cite{lee2023distributed} uses maximum likelihood estimation on Lie groups for mobile network self-tracking, but does not account for relative noises. An extension of SCI to Lie groups is proposed in~\cite{li2021joint}, though without an analysis of the algorithm's stability. Recently, Zheng \emph{et al.}~\cite{zhang2024towards} explores a novel direction by introducing a distributed invariant Kalman filter utilizing CI within mobile sensor networks for 2D target tracking under deception attacks but similarly neglects relative measurement noises.

We focus on a mobile multi-agent system where agents take relative pose and environmental measurements, with communication limited to neighboring agents. When they cooperatively estimate poses, the interplay between communication and relative observations introduces intricate correlations. We propose a Distributed Invariant CIKF (DInCIKF) for real-time distributed pose estimation, addressing unknown correlations with a focus on consistency and stability of estimation error covariance. 

The contributions of this paper are multifold. Primarily, we extend the invariant Kalman filter to a distributed framework based on covariance intersection, which mitigates the risk of producing over-conservative or over-confident estimates. By weaving in invariant error, we enhance the consistency and precision of pose estimation. Additionally, our adoption of a specialized pose measurement model ensures the system of logarithmic invariant error maintains linear characteristics, leading to desirable convergence properties. We further prove the consistency and stability of the DInCIKF, confirming that it's a solid choice for solving distributed pose estimation problems. 


The organization of this paper is as follows. Section II introduces the preliminaries. Section III defines the problem of our interest. Sections IV and V detail the DInCIKF algorithm and analyze its consistency and stability. Section VI presents simulations. Section VII concludes this work.

\section{Preliminaries}


\subsection{Kalman Filter and Covariance Intersection}~\label{CI}
Consistency is a vital characteristic for estimators. Conversely, an inconsistent estimate results in an overconfident estimate that underestimates actual errors. 
\begin{definition}[Estimation consistency~\cite{chang2021resilient}]
Consider a random vector $x \in \mathbb{R}^p$. Let $\hat x$ be an unbiased estimate of $x$ and $\hat P$ is the estimation error covariance. The estimate $\hat x$ is said to be consistent if 
$$
\mathbf{E}[(x-\hat x) (x-\hat x)^\top ] \leq \hat P,
$$
with the inequality defined in the sense of positive semi-definiteness.
\end{definition}

The KF assumes no correlation between the process noises and the measurement noises. The consistency of it is summarized in the following lemma:
\begin{lemma}[Consistency of KF]Consider uncorrelated, consistent and unbiased estimates of a random vector. The fused state estimate $\hat{x}_{\mathrm{KF}}$ with the estimation error covariance $\hat{P}_{\mathrm{KF}}$ given by KF is consistent, that is, $\mathbf{E}[(x-\hat{x}_{\mathrm{KF}})(x-\hat{x}_{\mathrm{KF}})^\top  ]\leq \hat{P}_{\mathrm{KF}}$. 
\label{l2}
\end{lemma}

CI is a conservative fusion approach that ensures consistent estimates without requiring knowledge of the exact correlations. Given estimates $\hat{x}_1,\ldots, \hat{x}_n$ of $x$ with estimation error covariance $\hat{P}_1,\ldots, \hat{P}_n$, CI fuses them as follows:
\begin{equation}~\label{CI}
\begin{aligned}
 P_{\mathrm{CI}}^{-1}=\sum_{i=1}^{n}\alpha_i \hat{P}_i^{-1}, \quad P_{\mathrm{CI}}^{-1} \hat{x}_{\mathrm{CI}} = \sum_{i=1}^n \alpha_i \hat{P}_i^{-1} \hat{x}_i,
\end{aligned}    
\end{equation}
where $0\leq \alpha_i \leq 1$ and $\sum_{i=1}^n\alpha_i=1$.
The coefficients $\alpha_i$ are determined by minimizing
$\operatorname{tr}(P_{\mathrm{CI}}),
$
which is a nonlinear problem and requires iterative search. In order to avoid the computational burden, a non-iterative coefficient-seeking algorithm was proposed in~\cite{niehsen2002information} as:
\begin{equation}~\label{eq:coeff}
    \alpha_i=\frac{1/\operatorname{tr}(\hat{P}_i)}{\sum_{i=1}^{n}1/\operatorname{tr}(\hat{P}_i)}.
\end{equation}

\begin{lemma}[Consistency of CI \cite{niehsen2002information}]The estimation error covariance $\hat{P}_{\mathrm{CI}}$ given by CI \eqref{CI} is consistent, that is,
$$
\mathbf{E}[(x-\hat{x}_{\mathrm{CI}})(x-\hat{x}_{\mathrm{CI}})^\top ] \leq \hat{P}_{\mathrm{CI}}.
$$
\label{consisCILemm}
\end{lemma}

\subsection{Matrix Lie Group and Lie Algebra}
Rigid body motion can be effectively modeled using matrix Lie groups. Three specific matrix Lie groups that will be used in this paper are introduced, collectively denoted by $G$. The special orthogonal group~$SO(3)$ is the set of rotation matrices in $\mathbb R^3$: $$
SO(3) \triangleq \{R\in \mathbb{R}^3| RR^\top  = I, \operatorname{det}(R)=1 \}.
$$
The special Euclidean group $SE(3)$, which comprises rotation and translation, is defined as
$$
 SE(3)\triangleq\left\{\begin{bmatrix}
R & p \\
0 & 1
\end{bmatrix}\bigg| R \in SO(3), p \in \mathbb{R}^3\right\}. $$
The matrix Lie group $SE_{2}(3)$ is as follows:
$$
SE_2(3)\triangleq\left\{ \left[\begin{array}{c|l}
R & p\quad v\\
\hline
0_{2 \times 3} &~~I_2 
\end{array}\right] \bigg| R \in SO(3), p, v \in \mathbb{R}^3\right\}.
$$

Each matrix Lie group is associated with a Lie algebra $\mathfrak{g}$, characterized by a vector space and a binary operation known as the Lie bracket. The Lie algebra of $SO(3)$ is given by:
$$
\mathfrak{so}(3) \triangleq \{ \omega ^\wedge \in \mathbb{R}^{3 \times 3} | \omega = [\omega_1~\omega_2~\omega_3]^\top  \in \mathbb{R}^3 \},
$$
where
$
\omega^{\wedge}=\begin{bmatrix}
0 & -\omega_3 & \omega_2 \\
\omega_3 & 0 & -\omega_1 \\
-\omega_2 & \omega_1 & 0
\end{bmatrix}
$. 

We abuse the notation $(\cdot)^\wedge$ in $SE(3)$ and $SE_2(3)$ to denote the mapping from the vector space to the corresponding Lie algebra. The inverse map of $(\cdot)^\wedge$ is denoted as map $(\cdot)^\vee$. The matrix form of Lie algebra associated with $SE(3)$ and $SE_2(3)$ are respectively given by:
$$
\begin{aligned}
&\mathfrak{s e}(3) \triangleq\left\{\left[\begin{array}{cc}
\omega^{\wedge} & t \\
{0}_{1 \times 3} & 0
\end{array} \right]\bigg| \omega^{\wedge} \in \mathfrak{s o}(3), t \in \mathbb{R}^3\right\},\\
&\mathfrak{se}_2(3) \triangleq\left\{\left[\begin{array}{c|c}{\omega}^{\wedge} & {t}_1 \quad {t}_2 \\ \hline {0}_{2 \times 3} & {0}_{2}\end{array}\right]\bigg| \omega^{\wedge} \in \mathfrak{s o}(3), t_1, {t}_2 \in \mathbb{R}^3\right\}.
\end{aligned}
$$

The exponential map is defined as $\operatorname{exp}(\xi)\triangleq \operatorname{exp}_m(\xi^\wedge): \mathbb{R}^n \rightarrow G$, where $\exp_m(\cdot)$ denotes the matrix exponential. The inverse mapping of exponential is the Lie logarithm, denoted by $\log(\cdot)$.

For $X \in G$ and $\xi \in \mathfrak{g}$, the adjoint map is as follows:
$$
\operatorname{Ad}_X : \mathfrak{g} \rightarrow \mathfrak{g};\quad  \xi ^ \wedge \mapsto \operatorname{Ad}_X (\xi^\wedge)\triangleq X \xi^\wedge  X^{-1}, 
$$
which can yield formulation:
\begin{align}~\label{eq:Ad_change_order}
  X \operatorname{exp}(\xi) =\operatorname{exp}(\operatorname{Ad}_X \xi ) X. 
\end{align}

Compounding two matrix exponentials via the Baker-Campbell-Hausdorff (BCH) formula is complex. For $x^{\wedge}, y^{\wedge} \in \mathfrak{g}$ and $\|y\|$ small, their compounded exponentials can be approximated to
\begin{align}
 \exp (x) \exp (y) \approx \exp (x + \operatorname{dexp}_{-x}^{-1} y),~\label{eq:exp_approx_smallbig} 
\end{align}
where $\operatorname{dexp}_x$ is the left Jacobian of $x$. In addition, if both $\|x\|$ and $\|y\|$ are small, it is simplified to
\begin{align}
\exp (x) \exp (y) \approx \exp (x+y).    ~\label{eq:exp_approx_small}
\end{align}

\section{Problem Formulation}
We consider a mobile multi-agent network comprising $n$ agents. Each agent employs a proprioceptive sensor Inertial Measurement Unit~(IMU), to generate self-motion information. Additionally, agents utilize exteroceptive sensors to acquire relative measurements against other agents and environmental features. The system is illustrated in Fig. \ref{problem}.


We adhere to the convention of using right subscripts for relative frames and left superscripts for base frames. The global frame is omitted if it is evident from the context.
\subsubsection{System Kinematics}
Let $p_i$, $R_i$ and $v_i$ denote the position, orientation and velocity of agent $i$  in the global frame. The state of each agent is represented as
$$
X_{i}\triangleq\left[\begin{array}{c|l}
R_i & p_i~~ v_i\\
\hline
0_{2 \times 3} &~~I_2 
\end{array}\right]\in SE_2(3).
$$
Agent $i$'s pose(orientation and position) is denoted by:
$$
T_i \triangleq \begin{bmatrix}
    R_i& p_i\\0_{1\times 3}&1
\end{bmatrix}.
$$

The 3D kinematics for agent $i$ is:
\begin{equation}~\label{eq:continuous_kinematics}
    \begin{aligned}
\dot{R}_i= R_i \omega_i^\wedge,\quad \dot{p}_i = v_i, \quad \dot{v}_i = a_i,
\end{aligned}
\end{equation}
where $w_i$, $a_i$ are angular velocity and acceleration, measured by the IMU. The measurements are modeled as:
$$
\omega_{i,m}=\omega_i+b_g+n_g,\quad a_{i,m} = R_i^\top  (a_i-g)+b_a+n_a,
$$
where the measurement noises $n_g$, $n_a$ are zero-mean white Gaussian noises, and $g$ denotes the gravity. The gyroscope and accelerometer bias $b_g$, $b_a$ are driven by white Gaussian noises $n_{b_g}$ and $n_{b_a}$, i.e., $ \dot{b}_g=n_{b_g}, \dot{b}_a = n_{b_a}
 $. The matrix differential equation for \eqref{eq:continuous_kinematics} is given by:
\begin{equation}~\label{eq:continuous_matrix_kinematics}
\dot{X}_i = X_i v_b +v_g X_i +MX_iN,
\end{equation}
where $M \triangleq \begin{bmatrix}
    I_3 & 0_{3 \times 2}\\0_{2 \times 3}& 0_2
\end{bmatrix}$, $N \triangleq \begin{bmatrix}
    0_{4\times 3}  & 0_{4 \times 2}\\0_{1\times 3}& 1 \quad 0
\end{bmatrix} $, and $v_b\triangleq \begin{bmatrix}
    (\omega_m-b_g-n_g)^\wedge & 0_{3\times 1} & a_m-b_a-n_a\\ 0_{2\times 3} & 0_{2\times 1} & 0_{2\times 1} 
\end{bmatrix} \in \mathfrak{se}_2(3)$, $v_g \triangleq \begin{bmatrix}
    0_3 & 0_{3 \times 1} & g\\
    0_{2\times 3} & 0_{2\times 1} &0_{2\times 1}
\end{bmatrix}\in \mathfrak{se}_2(3)$.

\subsubsection{Measurement Model}
The first type of exteroceptive sensors measurement model is the environmental measurements of features $f_s, s=1,2,\ldots$. The poses of environmental features are known and given:
$$
	T_{f_s}\triangleq\begin{bmatrix}
	R_{f_s}& p_{f_s}\\
	0&1\end{bmatrix} \in SE(3),
   \label{object_state}
$$
where $R_{f_s} \in SO(3)$ is the orientation of the feature in the global frame and $p_{f_s} \in \mathbb{R}^3$ is the feature's position. An environmental measurement is modeled as follows:
\begin{equation}~\label{eq:z_if}
^iT_{f_s,m}\triangleq h_{i f}(X_i)=T_i^{-1} \exp (n_{i f}) T_{f_s} ,
\end{equation}
where the measurement noise is modeled as white Gaussian noise $n_{if}\sim \mathcal{N}(0, R_{if}) $. 

The second type is the relative measurement between neighboring agents $i$ and $j$, which is modeled as:
\begin{equation}~\label{eq:z_ij}
^iT_{j,m}\triangleq h_{i j}(X_i, X_j)= T_i^{-1} \exp(n_{ij}) T_j , 
\end{equation}
where the relative measurement noise $n_{i j} \sim \mathcal{N}(0, R_{i j})$.

We remark that the poses of features can be obtained through semantic mapping techniques, and the relative pose measurements can be obtained from pose estimation methods, such as Perspective-n-Point~(PnP) in visual systems. The benefits stemming from the inclusion of these object-level measurements, denoting measurements related to features or neighboring agents as entire objects, are discussed in~\cite{9667208}. 
\begin{remark}
The uncertainty on the Lie algebra is typically modeled on the left (w.r.t. the global frame)~\cite{barfoot2014associating} as $^j T_{i,m}\triangleq \exp(n_{ij}) T_j^{-1} T_i$ or the right (w.r.t. the body-fixed frame)~\cite{lee2023distributed} as
$^j T_{i,m}\triangleq T_j(k)^{-1} T_i \exp(n_{ij})$ in the literature.
However, both of these models will introduce complex terms that depend on the system state into the estimator's error dynamics~\eqref{Pesti_cov}. Therefore, we opt to adopt the formulations in \eqref{eq:z_if} and \eqref{eq:z_ij}. 
\end{remark}

\subsubsection{Communication Network}
 The agent-to-agent communications are described by a directed graph $\mathcal{G}=(\mathcal{V}, \mathcal{E})$, with $\mathcal{V}$ representing the set of $n$ agents, and $\mathcal{E}$ the set of communication links. If $(j,i) \in \mathcal{E}$, agent $i$ can receive information from agent $j$. The set of neighbors of agent $i$ are denoted as $\mathcal{N}_i$.

Use a node $f$ to represent the set of environmental features and let $\mathcal{E}_{f}$ be a set of edges, where $(f,i) \in \mathcal{E}_{if}$ implies that node $i$ can access measurements concerning these features. Define an augmented graph $\bar{\mathcal{G}}=(\bar{\mathcal{V}}, \bar{\mathcal{E}})$, where $\bar{\mathcal{V}}=\mathcal{V} \cup \{f\}$, and $\bar{\mathcal{E}}=\mathcal{E}\cup {\mathcal{E}_{if}}$. We consider $\mathcal G$ and $\bar{\mathcal{G}}$ to be time-invariant and make the following assumptions.
\begin{assumption}
If agent $j$ is capable of measuring the relative pose of agent $i$, it can send information to agent $i$, that is, $(j,i)\in\mathcal E$.
\label{a1}
\end{assumption}
\begin{assumption}~\label{tree}
The augmented digraph $\bar{\mathcal{G}}$, contains a directed spanning tree rooted at node $f$.
\label{a2}
\end{assumption}

In multi-agent systems, environmental measurements and relative measurements are distributed across nodes. Exploring ways to harness the knowledge of individuals with abundant information to support those with limited information is crucial for system-wide precision and stability. In particular,
we are interested in \textit{designing a distributed pose estimator to compute $X_i$ for each agent $i$ based on the covariance intersection fusion scheme due to unknown noise inter-correlation among the information sources}.

\begin{figure}[H]
    \centering
    \includegraphics[width=0.85\linewidth]{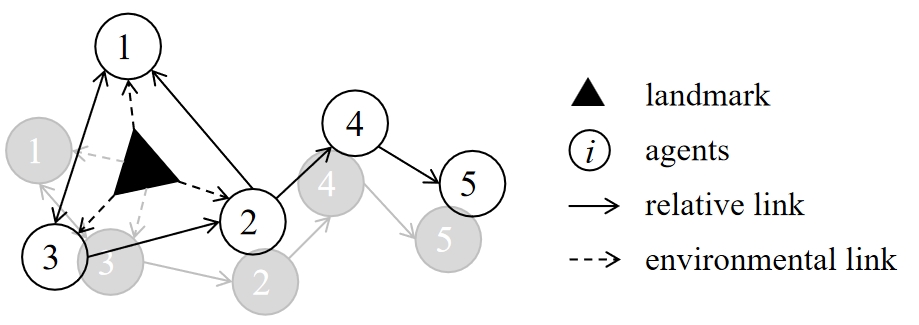}
     \caption{Multi-agent system with relative pose and environmental measurements.}
    \label{problem}
\end{figure}
\section{Distributed Invariant Kalman Filter based on Covariance Intersection}
In this section, we introduce the DInCIKF algorithm, which includes three steps: propagation step, the CI step and the KF update step. To proceed, we follow the convention of using the notation $\bar{(\cdot)}$ to denote the prior estimates after the prediction, $\breve{(\cdot)}$ to denote the estimates after the CI step, and $\hat{(\cdot)}$ to denote \textit{a posteriori} estimates the after KF update.

\subsection{Local Invariant Error Propagation}
Given the kinematics \eqref{eq:continuous_matrix_kinematics}, an estimate $\bar{X}_i$ to $X_i$ can be propagated as follows: 
\begin{equation}~\label{eq:X_prop_continuous}
    \dot{X}_i = X_i \hat{v}_b +v_g X_i +MX_iN,
\end{equation}
where $\hat{v}_b\triangleq \begin{bmatrix}
    (\omega_m-b_g)^\wedge & 0_{3\times 1} & a_m-b_a\\
    0_{2\times 3} & 0_{2\times 1}&0_{2\times 1}
\end{bmatrix}$\footnote{One can handle the estimation of the biases $b_a$ and $b_g$ of an IMU in conjunction with  the line of $SO(3)$ kinematics
under the iterated EKF framework, as outlined in~\cite{barrau2015non} and~\cite{hartley2020contact}. 
We
omit implementation details here. For those interested in the
finer points, we kindly refer them to the aforementioned two papers. }.

Define the right invariant error from $\bar{X}_i$ to $X_i$ as
$ \eta_i \triangleq X_i \bar{X_i}^{-1}$. The logarithmic right invariant error is defined as
$\xi_i =\operatorname{log}({\eta}_i)$. The dynamics of ${\xi}_i$ is as follows \cite{li2022closed}:
\begin{equation}\label{eq:xi_prop_continous}
    \dot{{\xi}}_i=F{\xi}_i+\operatorname{dexp}_{{\xi}_i}^{-1} \operatorname{Ad}_{\bar{X}_i} n_i,  
    \end{equation}
    where $n_i \triangleq [
    n_g^\top ~n_a^\top ~n_{b_g}^\top ~n_{b_a}^\top 
]^\top $ and $F \triangleq \begin{bmatrix}
    0_3 & 0_3 &0_3\\ 0_3& 0_3 & I_3 \\g^\wedge & 0_3 & 0_3    
    \end{bmatrix}$.

In order to implement the filtering algorithm using software to fuse the data from physical sensors, the discretization should be implemented \cite{hartley2020contact}. Through the discretization of the dynamics \eqref{eq:xi_prop_continous}, we obtain the dynamics of $\xi_i$:
\begin{equation}~\label{eq:xi_prop_discrete}
  \bar{\xi}_i(k+1)=A\hat{\xi}_i(k)+n_{i,d}, 
\end{equation}
where  $A \triangleq \begin{bmatrix}
    I & 0 & 0\\
    \frac{1}{2}g^\wedge \Delta t^2&I& I \Delta t\\ g^\wedge \Delta t &0&I
\end{bmatrix}$, $n_{i,d} \sim N(0_{9 \times 1}, Q_i)$ is the discrete approximation of the noise term in \eqref{eq:xi_prop_continous} obtained according to \cite{li2022closed} and  $\Delta t$ is the sampling interval. In this paper, we assume a constant $\Delta t$ so all agents have the same \(A\). Note that both $Q_i$ and $A$ are non-singular.

In virtue of \eqref{eq:xi_prop_discrete}, we use the following equation to approximate the evolution of the estimation error covariance:
\begin{equation}~\label{pred_P}
    \bar{P}_{i}(k+1)=A \hat{P}_i(k) A^\top +Q_i .
\end{equation}

\subsection{Invariant Covariance Intersection to Fuse Relative Measurements}
To facilitate agent $i$ in utilizing the information from its neighbors, it merges the relative measurements with the neighbors' self-estimates, subsequently fusing these with its prior estimates. This fusion introduces complicated correlations between the process noises and measurement noises, thus we use CI to give an consistent estimate. 

To be concise, we construct a virtual pose measurement of agent $i$ from \eqref{eq:z_ij}:
\begin{equation} 
\label{eq:z_ij_}
T_{i, m}^{(j)}(k) \triangleq \bar{T}_j(k)  ({}^j{T_{i, m}(k)}),
\end{equation}
Note that agent $j$ has local access to both $T_{i,m}^{(j)}(k)$ and $\bar{T}_j(k)$, thus $j$ is capable of calculating $T_{i,m}^{(j)}(k)$ locally and then broadcast it to agent $i$. $T_{i, m}^{(j)}(k)$ can be interpreted as a pose observation of agent $i$ from agent $j$, encompassing both the uncertainty of the relative measurement and the estimated uncertainty of agent $j$. Let $\delta$ denote the first six dimensions of $\xi$, and $\bar{\delta}_j \triangleq J  \bar{\xi}_j \in \mathbb{R}^6$, where
\begin{equation}\label{eqn:def_J}
J=[{I}_{6}~{0}_{6 \times 3}].
\end{equation}
Substitute \eqref{eq:z_ij} into \eqref{eq:z_ij_}:
$$
\begin{aligned}~\label{j_estoni}
T_{i, m}^{(j)}(k) & =\bar{T}_j(k) ( { }^j T_{i, m}(k)) \\
& =\exp (-\bar{\delta}_j(k))   T_j(k)    T_j(k)^{-1}  \exp (n_{i j})  \\
&=\exp (-\bar{\delta}_j(k))    \exp (n_{i j})    \  T_i(k)\\
&\overset{\eqref{eq:exp_approx_small}}{\approx} \exp (n_{i j}-\bar{\delta}_j(k))   T_i(k).
\end{aligned}
$$
Considering that $\mathbf{E}[\bar{\delta}_j(k) n_{ij}^\top ]=0$, the mean value of $(n_{ij}-\bar{\delta}_j(k))$ keeps zero since both $n_{ij}$ and $\bar{\delta}_j(k)$ have zero means. The uncertainty of $(n_{i j}-\delta_j(k))$ is:
\begin{align}
\tilde{R}_i^{(j)}(k)=  R_{i j} +J \bar{P}_{j}(k) J^\top .~\label{Pesti_cov}
\end{align}

Given that $\bar{P}_j(k)$ is predicted from $\hat{P}_j(k-1)$, and considering its dependency on neighbors, potentially including agent $i$, there exists a possibility that $\tilde{R}_i^{(j)}(k)$ and $\bar{P}_i(k)$ are correlated. Consequently, we apply CI \eqref{CI} to fuse $\bar{P}_i(k)$ with the $T_{i, m}^{(j)}(k)$ from agent $i$'s neighbors: 
\begin{subequations}~\label{CIupda}
\begin{align}
&\breve{P}_i^{-1}=\alpha_{i,k} \bar{P}_i^{-1}+\sum_{j \in \mathcal{N}_i} \alpha_{j,k} J^\top (\tilde{R}_i^{(j)})^{-1} J,\label{CIfuse1}\\
&\breve{\xi}_i=\breve{P}_i\sum_{j \in \mathcal{N}_i} J^\top  \alpha_{j,k}(\tilde{R}_i^{(j)})^{-1} \log (\bar{T}_i     ({{T}_{i,m}}^{(j)})^{-1}  )) \label{CIfuse2}, 
\end{align}
\end{subequations}
where $\alpha_{j,k}$ is determined by \eqref{eq:coeff} and we temporarily drop $(k)$ for brevity. $\bar{\xi}_i(k)$ it is omitted in \eqref{CIfuse2} for zero mean.

To implement KF update, we estimate the group state $\breve{X}_i(k)$ from $\breve{\xi}_i(k)$ and the prior estimate $\bar{X}_i(k)$:
\begin{align}~\label{mapping}
 \breve{X}_i(k)=\exp (\breve{\xi}_i(k)) \bar{X}_i(k).  
\end{align}

\subsection{Invariant Kalman Filtering Update to Fuse Environmental Measurements}
Given that the nature that $\mathbf{E}[\tilde{\xi}_i(k) n_{if}^\top ]=0$, to avoid an over-conservative estimate, we apply KF for the integration of environmental measurements.

We abuse the subscript $f$ to denote the features for simplicity. 
Similarly, we develop a virtual environmental measurement model that directly corresponds to the pose of agent $i$:
\begin{equation}~\label{virtual_zif}
 T_{i,m}^{(f)}(k) \triangleq T_f (^iT_{f,m}(k))^{-1}.   
\end{equation}
Substitute \eqref{eq:z_if} into \eqref{virtual_zif}:
$$
\begin{aligned}
T_{i,m}^{(f)}(k) &=T_f ( T_i(k)^{-1} \exp(n_{if})  T_f)^{-1}\\
&= T_f T_f^{-1} \exp(-n_{if}) T_i(k)  = \exp(-n_{if}) T_i(k).
\end{aligned}
$$
So the noise of $T_{i,m}^{(f)}(k)$ is $\delta_{i f} \sim \mathcal N(0, {R}_{i f})$. This result is intuitive for utilizing the known and given poses of features will not introduce additional uncertainty.

For simplicity, we omit the $(k)$ notation in the subsequent calculations. Define the residual as follows:
\begin{equation}~\label{res}
    r_i(\breve \xi_i)\triangleq \log (T_{i,m}^{(f)}   \breve{T}_i^{-1})=\log (T_{i,m}^{(f)} T_i^{-1}  \exp(J \breve{\xi}_i) ).
\end{equation}
The calculation of the Jacobian matrix is as follows:
\begin{align}
\frac{\partial r_i}{\partial \breve{\delta}_i}
&  =\lim _{\breve{\delta}_i \rightarrow 0} \frac{\log (T_{i ,m}^{(f)}   {T}_i^{-1}   \exp (\breve{\delta}_i))-\log (T_{i, m}^{(f)}   {T}_i^{-1})}{\breve{\delta}_i} \notag \\
& =\lim _{\breve{\delta}_i\rightarrow 0} \frac{\log (\exp (\phi_i)   \exp (\breve{\delta}_i))-\log (\exp (\phi_i))}{\breve{\delta}_i} \notag \\
& \overset{\eqref{eq:exp_approx_smallbig}}{\approx}\lim _{\breve{\delta}_i \rightarrow 0} \frac{\log (\exp (\phi_i+\operatorname{dexp}_{-\phi_i}^{-1}(\breve{\delta}_i)))-\phi_i}{\breve{\delta}_i} \notag  \\
&=\operatorname{dexp}_{-\phi_i}^{-1}
\triangleq J_{i f}, \quad \phi_{i} \triangleq \log (T_{i, m}^{(f)}   \breve{T}_i^{-1}) \notag.
\end{align}

Applying the chain rule, the Jacobian of the residual $r_i$ with respect to $\breve{\xi}_i$ is:
\begin{align}~\label{J}
  H_{i f}=\frac{\partial r_i}{\partial \breve{\xi}_i}=\frac{\partial r_i}{\partial \breve{\delta}_{i}}   \frac{\partial \breve{\delta}_{i}}{\partial \xi_i}=J_{i f}   J.
\end{align}
Subsequently, we use the standard KF update procedure:
\begin{equation}~\label{KFupdateP}
\begin{aligned}
& K_i=\breve{P}_i H_{i f}^{\top}(H_{i f} \breve{P}_i H_{i f}^{\top}+{R}_{i f})^{-1}, \\
& \hat{P_i}=(I-K_i H_{i f}) \breve{P}_i, \\
& \hat{\xi_i}=K_i r_i+\breve{\xi}_i.
\end{aligned}
\end{equation}
Finally, the process of mapping back to the Lie group to obtain the posterior estimate is as follows:
\begin{equation}~\label{KFupdateX}
    \hat{X}_i(k)=\exp (K_i(k)   r_i(k))   \breve{X}_i(k).
\end{equation}

To summarize, firstly, implement prediction according to IMU kinematics. Secondly, if the agent receives relative measurements, a CI update is performed. Finally, if environmental measurements are available, a KF update is executed. Each CI update requires neighbors broadcasting $T_{i,m}^{(j)}(k)$ to agent $i$. This procedure is concluded in Algorithm \ref{algorithm}.

\begin{algorithm} 
	\caption{Distributed Invariant CIKF~(DInCIKF)} 
	\label{algorithm} 
	\begin{algorithmic}
		\REQUIRE $\hat{X}_i(k-1), \hat{P}_i(k-1), ^iT_{f,m}(k), ^iT_{j,m}(k),$
        \STATE $~~~~\quad \quad \omega_m(k), a_m(k), b_a, b_g$.
		\ENSURE $\hat{X}_i(k), \hat{P}_i(k)$.
		\STATE \textbf{\underline{Step 1: Prediction.}} 
 \STATE Propagate $\bar{X}_i(k)$ and $\bar{P}_i(k)$ using \eqref{eq:X_prop_continuous} and \eqref{pred_P}.
        \STATE \textbf{\underline{Step 2: CI update.}}
		\STATE (1) Compute $T_{i,m}^{(j)}(k)$ and $\tilde{R}_i^{(j)}(k)$ using \eqref{eq:z_ij_} and \eqref{Pesti_cov};
        \STATE (2) Update $\breve{\xi}_i(k)$, $\breve{P}_i(k)$ using \eqref{CIupda}, and $\breve{X}_i(k)$ using \eqref{mapping}.
        \STATE \textbf{\underline{Step 3: KF update.}}
		\STATE (1) Compute $r_i(k)$ and $H_{if}$ using (\ref{res}) and (\ref{J});
        \STATE (2) Update $\hat{P}_i(k)$ and $\hat{X}_i(k)$ 
 using (\ref{KFupdateP}) and (\ref{KFupdateX}).
	\end{algorithmic} 
\end{algorithm}

\section{Theoretical Analysis}
\subsection{Consistency Analysis of DInCIKF}
We first make the following assumption on the initialization of DInCIKF:
\begin{assumption}\label{assumption:consistecy_initial_P}
The initial estimation error covariance $\hat{P}_i(0)$ is consistent.   
\label{a3}
\end{assumption}
We establish the following theorem.
\begin{theorem}[Consistency of DInCIKF]\label{thm:consistency} DInCIKF is consistent under Assumption~\ref{a3}, i.e., $\mathbf{E}[\hat{\xi}_i(k)\hat{\xi}_i^\top (k)]\leq \hat{P}_i(k)$ for any $i\in\mathcal{V}$ and $ k\ge0$.
\end{theorem}
\begin{proof}
We prove the result by induction. 

At $k =0$, the result holds due to Assumption \ref{a3}. Now assume $\mathbf{E}[\hat{\xi}_i(k-1) \hat{\xi}_i(k-1)^\top ] \leq \hat{P}_i(k-1)$ at $k-1$. Then, since $\mathbf{E}[\bar{\xi}_i(k)  (n_{d})^\top ] =0$, $$
\begin{aligned}
    \mathbf{E}[\bar{\xi}_i(k) \bar{\xi}_i(k)^\top  ] &= A (\mathbf{E}[\hat{\xi}_i(k-1) \hat{\xi}_i(k-1)^\top  ) ]A^\top  + Q_i\\ &\leq A \hat{P}_i(k-1)A^\top  + Q_i = \bar{P}_i(k).
\end{aligned}
$$
By Lemma \ref{consisCILemm}, after the CI step, it follows that
$$\mathbf{E}[\breve{\xi}_i(k)(\breve{\xi}_i(k))^\top ] \leq \breve{P}_i(k).$$
By Lemma \ref{l2} and $\mathbf{E}[\breve{\xi}_i(k)(n_{if})^\top ]=0$, the KF update yields
$\mathbf{E}[(\hat{\xi}_i(k))(\hat{\xi}_i(k))^\top ]\leq \hat{P}_i(k)$, which completes the proof.
\end{proof}

\subsection{Stability Analysis}
Under a mild connectivity assumption on the augmented graph $\bar{\mathcal G}$, we can show that the DInCIKF is stable in terms of expected squared estimation error over time.
\begin{theorem}[Stability of DInCIKF]\label{thm:stability_DInCIKF}
Under Assumptions~\ref{a2} and~\ref{assumption:consistecy_initial_P}, the proposed DInCIKF is stable in the sense that $\mathbf E[\hat{\xi}(k)\hat{\xi}(k)^\top ]$ is bounded across time $k$ for all $i\in\mathcal V$.
\end{theorem}

The proof of Thoerem~\ref{thm:stability_DInCIKF} is constructive by showing the existence of an 
upper bound for the error covariance of DInCIKF. Its outline is as follows. First, we construct an Auxiliary Upper Bound System (AUBS), which is claimed in Lemma \ref{construct_AUBS}. Then we prove the convergence of the AUBS state by Lemma \ref{lamammamamamam} and Lemma \ref{CILINK}. Finally, we conclude the proof by analyzing the observability of the error dynamics of the proposed filter.
Some technical lemmas that support the proof are provided in the Appendix. 

\textit{Proof of Theorem~\ref{thm:stability_DInCIKF}.} After the CI and KF update, the \textit{a posteriori} estimation error covariance is given by
\begin{align}~\label{combine}
&\hat{P}_i(k)^{-1}\notag\\
= &\alpha_{i,k} \bar{P}_i(k)^{-1}+H_{if}^\top  {R}_{if}^{-1} H_{if}\notag\\
&+\sum_{j \in \mathcal{N}_i(k)} \alpha_{j,k} J^\top  ( {R}_{ij}+  J \bar{P}_{j}(k) J^\top        )^{-1} J \notag\\
 =& \alpha_{i,k} \bar{P}_i(k)^{-1}+H_{if}^\top  {R}_{if}^{-1} H_{if}\\&+\sum_{j \in \mathcal{N}_i(k)} \alpha_{j,k} J^\top  ( {R}_{ij}+  J (A\hat{P}_j(k-1)A^\top  +Q_d) J^\top        )^{-1} J.\notag
\end{align}

Now we construct a directed spanning tree from the augmented graph $\bar{\mathcal G}$ and thereby construct an AUBS for each agent $i$. 
With Assumption~\ref{tree}, we can construct a directed spanning tree from with the root being node $f$, denoted as $\bar{\mathcal{T}}=( \bar{\mathcal{V}}, \mathcal{E}_T)$. The set $\bar{\mathcal{V}}$ 
 is partitioned as $\bar{\mathcal{V}} = \mathcal{V}_1\cup
\mathcal{V}_2\cup\cdots$, 
where $\mathcal{V}_1$ is the set of all child nodes of  root $f$ and $\mathcal{V}_2$ is the set of all child nodes of the nodes in $\mathcal V_1$, and so forth. The AUBS is constructed as follows:
\begin{equation}~\label{AUBS}
\begin{aligned}
& \bar{\Pi}_i(k)=A\hat{\Pi}_i(k-1)A^\top +Q_d, \\
& \hat{\Pi}_i(k)^{-1}=\alpha_{i}\bar{\Pi}_i(k)^{-1}+\check{H}_i^{T}\check{R}_i^{-1} \check{H}_i,
\end{aligned}
\end{equation}
initialized with 
$\bar{\Pi}_{i}(0) \geq \bar{P}_{i}(0)$, where $n_i$ denotes the number of elements in $\mathcal{N}_i \cup \{i\}$. 
In~\eqref{AUBS}, if $i \in \mathcal{V}_1$, $\check{H}_i \triangleq H_{if}$, and $\check{R}_i \triangleq R_{if}$. Provided that a node $j$'s $\check{H}_j$ and 
$\check{R}_j$ have been defined, then, for any its child node $i$ in $\bar{ \mathcal T}$, $\check{H}_i \triangleq \check{H}_j(A^{-1} J^\top J)$ and $\check{R}_i \triangleq c_j \check{R}_j$, where $c_j$ is calculated following~\eqref{eqn:def_c} in Appendix and $\alpha_i = \min(\alpha_{i,k})$. We can recursively define $\check{H}_i$ and 
$\check{R}_i$ for all the nodes across the network.

We next verify the convergence of AUBS. We first introduce the definition of the observability matrix.
\begin{definition} The observability matrix of $(A,\check H_i)$ is
$$O(A,\check{H}_i) = [(\check{H}_i)^\top ~(\check{H}_iA)^\top ~\ldots~(\check{H}_iA^{8})^\top ]^\top ,
$$
where $A \in \mathbb{R}^{9\times 9}$. The system $(A,\check{H}_i)$ is said to be observable when $O(A,\check{H}_i)$ has full column rank.
\end{definition}

To proceed, we need the following lemma that is grounded in optimal control theory~\cite{kailath2000linear}. The proof is omitted due to space constraints.
\begin{lemma}
Given that $({A}, Q_i^{1/2})$ is controllable and $({A}, \check{H}_i)$ is observable, with $\Pi_i(0) \geq 0$, the recursion of   \eqref{AUBS}:
$$
\bar{\Pi}_i(k+1)=\check{A}(\bar{\Pi}_i(k)^{-1} +\check{H}_i^\top  (\alpha_i \check{R}_i)^{-1} \check{H}_i    )^{-1} \check{A}^\top  +Q_i,
$$
converges to the $\Pi_i$, which is a unique PD solution to the discrete algebraic Riccati equation
$$
\Pi_i = \check{A}_i \Pi_i \check{A}_i^\top  +Q_i - \Pi \check{H}_i^\top  (\alpha_i \check{R}_i + \check{H}_i^\top  \Pi \check{H}_i)^{-1}\check{H_i}\Pi,
$$
where $\check{A}_i \triangleq \frac{1}{\sqrt{\alpha_i}}A$.
\label{lamammamamamam}
\end{lemma}

Note that the nonsingularity of $Q_i$ easily verifies the controllability. To show the convergence, it suffices to show 
the observability of $(A,\check{H}_i)$, which is indicated in the Lemma~\ref{CILINK} in Appendix. 
By Lemma~\ref{CILINK}, since all the nodes in $\bar{\mathcal{T}}$ have access to node $f$, the AUBS is observable for all $i\in \mathcal V$. The convergence of 
$\bar{\Pi}_i(k)$ follows by 
Lemma \ref{lamammamamamam}.

The following lemma connects DInCIKF and the AUBS, the proof of which is provided in the Appendix.  
\begin{lemma}~\label{construct_AUBS}
$
\hat{\Pi}_i(k) \geq \hat{P}_i(k)$ and $    \bar{\Pi}_i(k) \geq \bar{P}_i(k), ~ \forall k \geq 0.
$
\end{lemma}

We finally conclude the proof by invoking consistency of DInCIKF from Theorem~\ref{thm:consistency}.
$\blacksquare$

\section{Simulation}
In this section, we present the simulation results of distributed pose estimation for a multi-agent system with five agents. Each individual is equipped with IMU providing line acceleration and angular measurements at 100Hz, and generates relative measurements and environmental measurements at 50Hz. We test two scenarios as illustrated in Fig. \ref{f2} and Fig. \ref{f3}, while the two scenarios share the same communication topology in Fig.\ref{com}. 
\begin{figure}[H]
    \centering
    \includegraphics[width=0.35\linewidth]{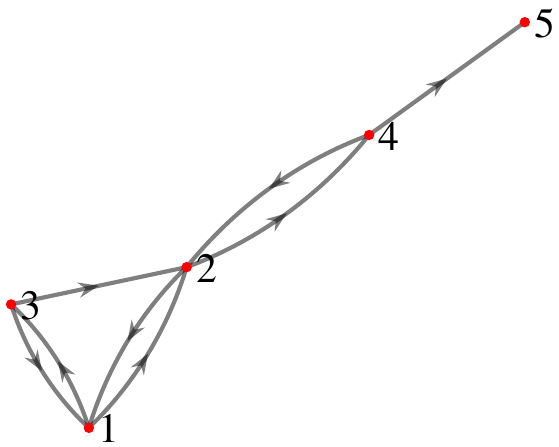}
    \caption{The communication network setting for simulation.}
    \label{com}
\end{figure}

\begin{figure}[H]
    \centering
    \includegraphics[width=0.9\linewidth]{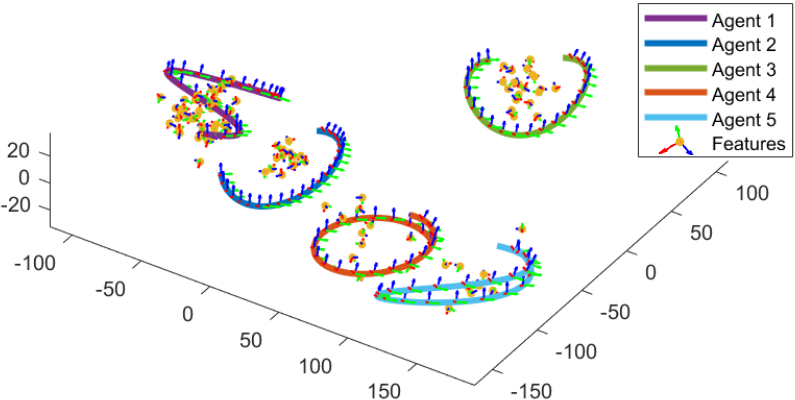}
    \caption{In Scenario 1, there are five agents with (1) different
  intensities of measurement noises, (2) a different number of features. This scenario is designed to simulate whether an agent with 
abundant information to support those with limited information in terms of 
estimation performance.}
    \label{f2}
\end{figure}
\begin{figure}[H]
    \centering
    \includegraphics[width=0.9\linewidth]{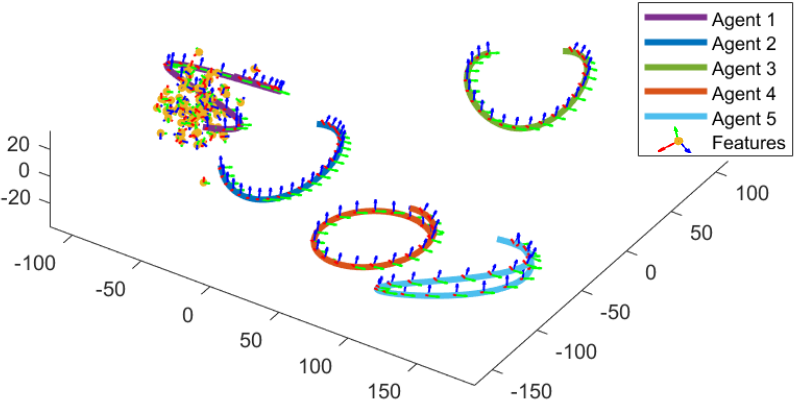}
    \caption{In Scenario 2, an extreme case is tested to showcase the effectiveness of the proposed algorithm. 
     It comprises three agents (Agent 3, 4, 5) who are completely blind, one agent (Agent 2) intermittently sensing features, and one agent (Agent 1) continuously observing features all the time. The intensities of the measurement noises for the agents are set to be the same. }
    \label{f3}
\end{figure}
We compare the proposed algorithm (DInCIKF) with standard $SO(3)$-Extended Kalman Filter~(Std-EKF) \cite{9667208} and Invariant Extended Kalman Filter~(InEKF)~\cite{li2022closed}. We use the 

The evaluation metric, following \cite{chang2021resilient}, is the average root-mean-squared-error~(RMSE) of rotation and position among five agents:
$$
\begin{aligned}
\operatorname{RMSE}_{rot}(k)&=\sqrt{\frac{ \sum_{i=1}^5\| \log(\hat R_i(k)^\top  R_i(k))   \|^2 }{5}},
\\
\operatorname{RMSE}_{pos}(k)&=\sqrt{\frac{ \sum_{i=1}^5\| \hat p_i(k)- p_i(k)  \|^2 }{5}}.
\end{aligned}
$$

The estimation performance results in scenario 1 are presented in Fig. \ref{Result1}. All methods successfully estimate the poses with the observation of features. Among these methods, DInCIKF stands out by appropriately fusing the correlated and the independent measurements. The results in Fig. \ref{Result2} demonstrate the failure of distributed pose estimation for Std-EKF in an extreme environment. Especially for agent 5, which is blind and entirely relies on neighboring agent 4 as indicated in Fig. \ref{com}, which is also blind to the features. Without reliable fusion, agent 5 will introduce estimation error to the whole system. Meanwhile, the accurate estimates provided by DInCIKF indicate its superior when solving distributed pose estimation problems.
\begin{figure}[H]
    \centering
    \includegraphics[width=0.91\linewidth]{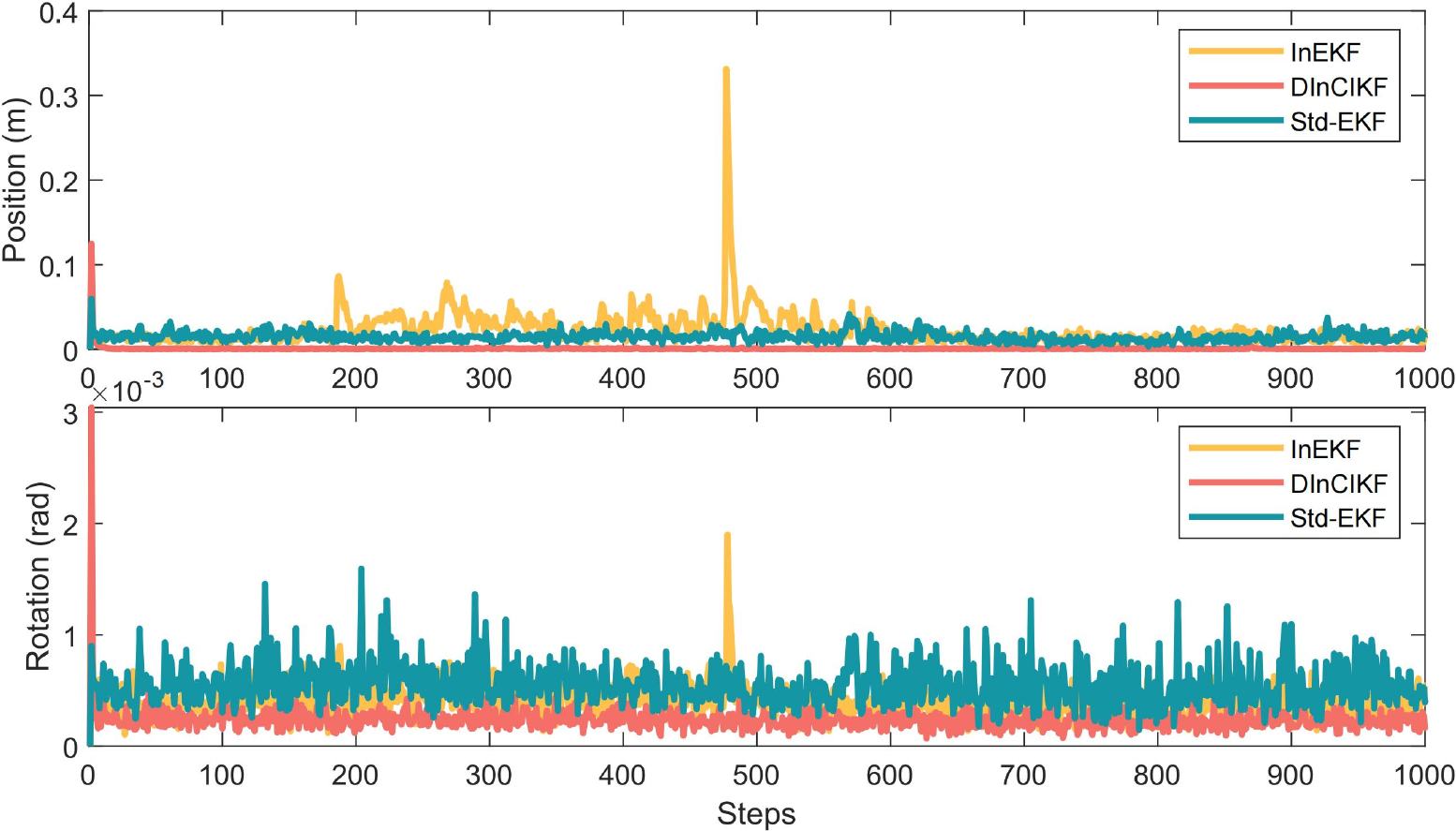}
    \caption{RMSE evolution of scenario 1.}
    \label{Result1}
\end{figure}
\begin{figure}[H]
    \centering
    \includegraphics[width=0.9\linewidth]{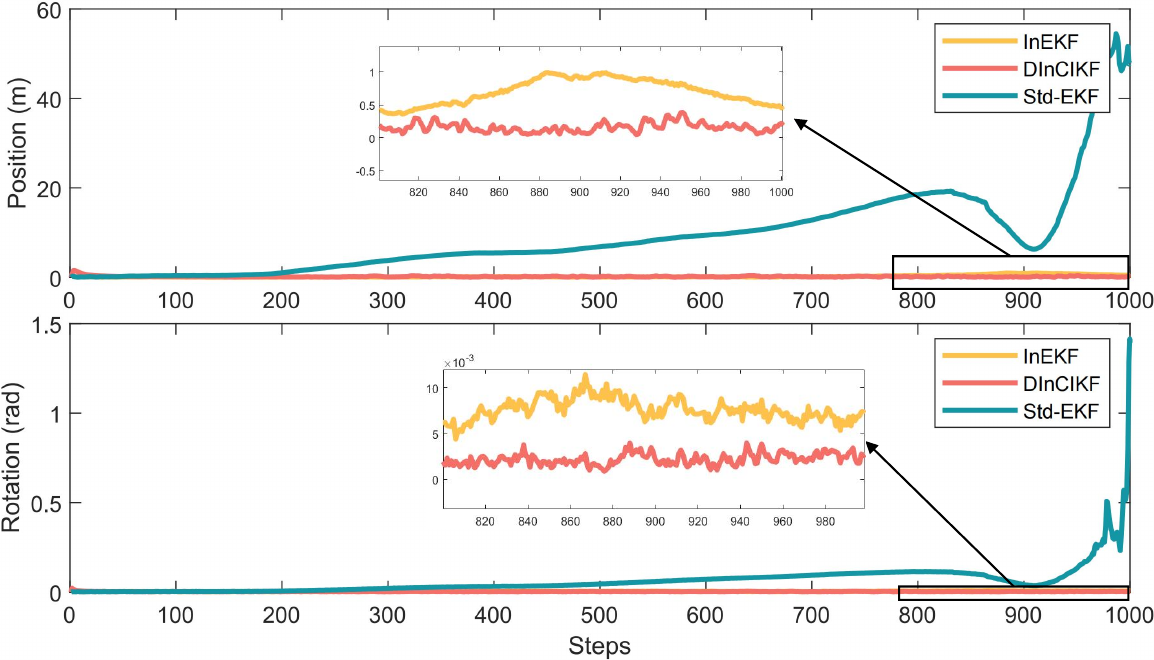}
    \caption{RMSE evolution of scenario 2.}
    \label{Result2}
\end{figure}

\section{Conclusion and Future Work}
We proposed the DInCIKF algorithm to address the distributed pose estimation problem in a mobile multi-agent network. We also analyzed the consistency and stability of DInCIKF to provide a theoretical guarantee for DInCIKF. Through simulations, we validated the algorithm's accuracy and the significance of motion representation on Lie groups. 
Moving forward, we will pursue better integration of CI and KF, which relies on closer correlation analysis for the information sources. We will also consider time-varying systems in the future, such as the scenarios involving intermittent communications.

\section{Appendix: Technical Support for Proof of Theorem~\ref{thm:stability_DInCIKF}}
\subsection{Supporting Lemmas}
\begin{lemma}~\label{l3}
For a PD matrix $M\in \mathbb{R}^{p\times p}$ and an integer $q$ satisfying $0<q \leq p $, there exists a scalar $0<\mu<1$ such that $([M]_{1:q,1:q})^{-1} \geq \mu [M^{-1}]_{1:q,1:q}  $.
\end{lemma}
\begin{proof}
    Let $M \triangleq \begin{bmatrix}
        M_{11} & M_{12}\\M_{21} & M_{22}
    \end{bmatrix}$, where $M_{11} \triangleq [M]_{1:q,1:q} \in \mathbb{R}^{q \times q}$. By the matrix inversion lemma, we have $[M^{-1}]_{1:q,1:q} = (M_{11}- M_{12}M_{22}^{-1}M_{21})^{-1}.$ Since
$M_{11} \geq M_{11}-M_{12} M_{22}^{-1}M_{21}$ and both of them are PD, we can choose $0<\mu<1$, such that $$
    M_{11}^{-1/2}(M_{11}- M_{12}M_{22}^{-1}M_{21}) M_{11}^{-1/2} \geq \mu I_{q},
    $$ which completes the proof. 
\end{proof}

\begin{lemma}~\label{RQ_bound}
There exists a scalar  $\sigma_j>0$ such that    
    \begin{equation}
        R_{ij}+J Q_d J^\top  \leq \sigma_j I,~\label{L5-1}
    \end{equation}
where $J$ is given in~\eqref{eqn:def_J}.
\end{lemma}
\begin{proof}
The proof is straightforward and, therefore, skipped. 
\end{proof}

\begin{lemma} The system $(A,{H}_{if}(A^{-1}J^\top  J)^m)$ is observable for any $ m \in \mathbb{N}$. 
\label{CILINK}
\end{lemma}
\begin{proof}
Given that $H_{if}=J_{if}J$ and $J_{if}$ being invertible, $$O(A, H_{if}(A^{-1}J^\top J)^m) = J_{if}O(A, J(A^{-1}J^\top J)^m),$$ and consequently $O(A,H_{if}(A^{-1}J^\top J)^m)$ has the same rank with $O(A, J(A^{-1}J^\top J)^m)$. So if $O(A, J(A^{-1}J^\top J)^m)$ has full column rank, then  $(A,{H}_{if}(A^{-1}J^\top  J)^m)$ is observable.

For each $m\geq 0$, one can calculate the matrix
$$\begin{bmatrix}
J(A^{-1}J^\top J)^m\\
J(A^{-1}J^\top J)^m A
\end{bmatrix}  =  \begin{bmatrix}
    I_3 & 0_3 & 0_3 \\
    \frac{(m)\Delta t^2}{2} g^\wedge & I_3 & 0_3\\
    I_3 & 0_3 & 0_3\\
    \frac{(m+1) \Delta t^2}{2} g^\wedge  & I_3 & \Delta t I_3,
\end{bmatrix}$$
which has full column rank, indicating that the matrix $O(A, J(A^{-1}J^\top J)^m)$ is full column rank. 
\end{proof}

\subsection{Proof of Lemma~\ref{construct_AUBS} }

First of all, consider $l \in \mathcal{V}_1$, $\check{H}_l = H_{lf}$, $\check{R}_l = H_{lf}$.
At $k=0$, $\bar{\Pi}_l(0) \geq \bar{P}_i(0)~\hat{\Pi}_l(0) \geq \hat{P}_l(0)$ by Assumption~\ref{assumption:consistecy_initial_P}.
Given the hypothesis that $\bar{\Pi}_l(k-1) \geq \bar{P}_l(k-1), \hat{\Pi}_l(k-1) \geq \hat{P}_l(k-1)$, then
\begin{align*}
\bar{\Pi}_l(k)& =A \hat{\Pi}_l(k-1)A^\top  +Q_l \notag \\
&\geq A \hat{P}_l(k-1)A^\top  +Q_l  = \bar{P}_l(k),
\end{align*}
which further yields that
\begin{align*}
&\hat{P}_l(k)^{-1}\notag\\
 =& \alpha_{l,k} \bar{P}_l(k)^{-1}+H_{lf}^\top  {R}_{lf}^{-1} H_{lf}\\
 &+\sum_{j \in \mathcal{N}_l(k)} \alpha_{j,k} J^\top  ( {R}_{lj}+  J (A\hat{P}_j(k-1)A^\top  +Q_j) J^\top        )^{-1} J\notag\\
 \geq & \alpha_{l,k} \bar{P}_l(k)^{-1}+H_{lf}^\top  {R}_{lf}^{-1} H_{lf}\notag\\
 \geq & \alpha_{l}\bar{\Pi}_l(k)^{-1}+H_{lf}^\top  {R}_{lf}^{-1} H_{lf}=\hat{\Pi}_l(k)^{-1}.\notag
\end{align*}
indicating that $\hat{\Pi}_l(k)\geq\hat{P}_l(k)$.

Next, we will show the results for  
 $i \in \mathcal{V}_{2}$. Note that we have shown for any $j\in\mathcal V_1$:
\begin{equation}~\label{eq:Pj_bound}
\bar{\Pi}_j(k)\geq \bar{P}_j(k),~\hat{\Pi}_j(k) \geq \bar{P}_j(k), ~\forall k \geq 0    
\end{equation}
Suppose that $i$ is a child node of $j$,  At $k=0$, $\hat{\Pi}_i(0) \geq \hat{P}_i(0)$ and $\bar{\Pi}_i(0) \geq \bar{P}_i(0)$ with appropriate initialization.
Given the hypothesis $\hat{\Pi}_i(k-1) \geq \hat{P}_i(k-1)$ and $\bar{\Pi}_i(k-1) \geq \bar{P}_i(k-1)$. At time $k$, by discarding local observation terms and retaining only one neighboring node  $j \in \mathcal{V}_1$ of agent $i$,
\begin{align}
\hat{P}_i(k)^{-1}=&\alpha_{i,k} \bar{P}_i(k)^{-1}+H_{if}^\top  {R}_{if}^{-1} H_{if}\notag\\
&+\sum_{s \in \mathcal{N}_i(k)} \alpha_{s,k} J^\top  ( {R}_{is}+  J \bar{P}_{s}(k) J^\top        )^{-1} J\notag\\
\geq &  \alpha_{i,k} \bar{P}_i(k)^{-1}+  \alpha_{j,k} \Omega_j(k),~\label{bound_eq1}
\end{align}
where
$$\Omega_s(k)\triangleq J^\top  ( {R}_{is}+  J (A\hat{P}_s(k-1)A^\top  +Q_s) J^\top  )^{-1} J.$$ 
Based on \eqref{eq:Pj_bound} and Lemma \ref{RQ_bound}, a lower bound of $\Omega_j(k)$ is derived as follows:
\begin{align}
&\Omega_j(k)\notag\\
\geq & J^\top  ( \sigma_j I+  JA\hat{P}_j(k-1)A^\top J^\top  )^{-1} J\notag \\
{\geq} & J^\top (\sigma_j I +JA\hat{\Pi}_j(k-1)A^\top J^\top )^{-1}J. \notag \\
=&J^\top (\sigma_j I +JA(\alpha_j\bar{\Pi}_j(k-1)^{-1} + \check{H}_j^\top \check{R}_j^{-1}\check{H}_j)^{-1})A^\top J^\top )^{-1}J. \notag
\end{align}
Since we have shown the convergence of $\bar{\Pi}_j(k)$, 
there exists a $\beta_j'>0$ such that 
   $\bar{\Pi}_j(k-1)^{-1} \geq \beta_j' I$ and hence 
\begin{equation*}
\Omega_j(k)\geq J^\top (\sigma_jI +JA(\beta_j I+\check{H}_j^\top  \check{R}_j^{-1}\check{H}_j)^{-1})A^\top J^\top )^{-1}J, 
\end{equation*}
with $\beta_j=\alpha_j \beta_j'$.
Furthermore,  we can find a scalar $\gamma_j'>0$ satisfying 
    \begin{align*}
       \gamma_j' I \leq   JA(\beta_j I + \check{H}_j ^\top \check{R}_j^{-1}\check{H}_j)^{-1})A^\top J^\top .
    \end{align*}
Combing all of these, by letting $\gamma_j \triangleq \frac{\gamma_j'}{\sigma_j+\gamma_j'}$, it follows that:
\begin{align*}
&(\sigma_jI +JA(\beta_j I+\check{H}_j ^\top \check{R}_j^{-1}\check{H}_j)^{-1})A^\top J^\top )^{-1}\notag \\
\geq & \gamma_j (JA(\beta_j I+\check{H}_j ^\top \check{R}_j^{-1}\check{H}_j)^{-1})A^\top J^\top )^{-1}.
\end{align*}
Then we have 
\begin{equation*}
\Omega_j(k) \geq \gamma_j J^\top  (JA (\beta_j I + \check{H}_j^\top  \check{R}_j^{-1}  \check{H}_j     )^{-1} A^\top   J^\top )^{-1}J.
\end{equation*}
By Lemma~\ref{l3}, there exist $0<\mu_j<1$ such that
\begin{align}
&\Omega_j(k)\notag \\
\geq &\gamma_j \mu_j J^\top J A^{-T}(\beta_j I + \check{H}_j ^\top \check{R}_j^{-1}\check{H}_j) A^{-1}J^\top J\notag\\
\geq &\gamma_j \mu_j J^\top J A^{-T}(\check{H}_j ^\top \check{R}_j^{-1}\check{H}_j) A^{-1}J^\top J. ~\label{Omega_inequality}
\end{align}
Substituting \eqref{Omega_inequality} into \eqref{bound_eq1}, it follows that 
\begin{align}~\label{combine}
\hat{P}_i(k)^{-1}
 \geq& \alpha_{i,k} \bar{P}_i(k)^{-1} + \alpha_{j,k} \Omega_j(k)\notag\\
\geq &\alpha_{i} \bar{\Pi}_i(k)^{-1} \notag\\
&\quad +  \alpha_{j,k}\gamma_j \mu_j J^\top J A^{-T}(\check{H}_j^\top  \check{R}_j^{-1}\check{H}_j) A^{-1}J^\top J\notag\\
= &\alpha_{i} \bar{\Pi}_i(k)^{-1} +  \check{H}_i^\top  \check{R}_i^{-1}\check{H}_i = \hat{\Pi}_i(k)^{-1},\notag
\end{align}
where $\check{H}_i = \check{H}_j A^{-1}J^\top J$ and $\check{R}_i = c_j \check{R}_j$ with 
\begin{equation}\label{eqn:def_c}
c_j\triangleq \alpha_{j}\gamma_j \mu_j.
\end{equation}
Therefore, $\hat{\Pi}_i(k) \geq \hat{P}_i(k)$
for $i\in\mathcal V_2$.

We repeat the above procedure to show $\hat{\Pi}_i(k) \geq \hat{P}_i(k)$
for $i\in \mathcal{V}_{3}$, and so on and so forth, which completes the proof. 

\bibliographystyle{IEEEtran}
\bibliography{CDC_ref}

\begin{thebibliography}{10}
\providecommand{\url}[1]{#1}
\csname url@samestyle\endcsname
\providecommand{\newblock}{\relax}
\providecommand{\bibinfo}[2]{#2}
\providecommand{\BIBentrySTDinterwordspacing}{\spaceskip=0pt\relax}
\providecommand{\BIBentryALTinterwordstretchfactor}{4}
\providecommand{\BIBentryALTinterwordspacing}{\spaceskip=\fontdimen2\font plus
\BIBentryALTinterwordstretchfactor\fontdimen3\font minus \fontdimen4\font\relax}
\providecommand{\BIBforeignlanguage}[2]{{%
\expandafter\ifx\csname l@#1\endcsname\relax
\typeout{** WARNING: IEEEtran.bst: No hyphenation pattern has been}%
\typeout{** loaded for the language `#1'. Using the pattern for}%
\typeout{** the default language instead.}%
\else
\language=\csname l@#1\endcsname
\fi
#2}}
\providecommand{\BIBdecl}{\relax}
\BIBdecl

\bibitem{dorri2018multi}
A.~Dorri, S.~S. Kanhere, and R.~Jurdak, ``Multi-agent systems: A survey,'' \emph{Ieee Access}, vol.~6, pp. 28\,573--28\,593, 2018.

\bibitem{trawny2005indirect}
N.~Trawny and S.~I. Roumeliotis, ``Indirect kalman filter for 3d attitude estimation,'' \emph{University of Minnesota, Dept. of Comp. Sci. \& Eng., Tech. Rep}, vol.~2, p. 2005, 2005.

\bibitem{barrau2016invariant}
A.~Barrau and S.~Bonnabel, ``The invariant extended kalman filter as a stable observer,'' \emph{IEEE Transactions on Automatic Control}, vol.~62, no.~4, pp. 1797--1812, 2016.

\bibitem{carli2008distributed}
R.~Carli, A.~Chiuso, L.~Schenato, and S.~Zampieri, ``Distributed kalman filtering based on consensus strategies,'' \emph{IEEE Journal on Selected Areas in communications}, vol.~26, no.~4, pp. 622--633, 2008.

\bibitem{chen2002estimation}
L.~Chen, P.~O. Arambel, and R.~K. Mehra, ``Estimation under unknown correlation: Covariance intersection revisited,'' \emph{IEEE Transactions on Automatic Control}, vol.~47, no.~11, pp. 1879--1882, 2002.

\bibitem{li2013cooperative}
H.~Li and F.~Nashashibi, ``Cooperative multi-vehicle localization using split covariance intersection filter,'' \emph{IEEE Intelligent transportation systems magazine}, vol.~5, no.~2, pp. 33--44, 2013.

\bibitem{7954723}
Z.~Wu, Q.~Cai, and M.~Fu, ``Covariance intersection for partially correlated random vectors,'' \emph{IEEE Transactions on Automatic Control}, vol.~63, no.~3, pp. 619--629, 2018.

\bibitem{chang2021resilient}
T.-K. Chang, K.~Chen, and A.~Mehta, ``Resilient and consistent multirobot cooperative localization with covariance intersection,'' \emph{IEEE Transactions on Robotics}, vol.~38, no.~1, pp. 197--208, 2021.

\bibitem{zhu2020fully}
P.~Zhu and W.~Ren, ``Fully distributed joint localization and target tracking with mobile robot networks,'' \emph{IEEE Transactions on Control Systems Technology}, vol.~29, no.~4, pp. 1519--1532, 2020.

\bibitem{10232373}
J.~Xu, P.~Zhu, Y.~Zhou, and W.~Ren, ``Distributed invariant extended kalman filter using lie groups: Algorithm and experiments,'' \emph{IEEE Transactions on Control Systems Technology}, vol.~31, no.~6, pp. 2777--2789, 2023.

\bibitem{lee2023distributed}
J.-G. Lee, Q.~Van~Tran, K.-H. Oh, P.-G. Park, and H.-S. Ahn, ``Distributed object pose estimation over strongly connected networks,'' \emph{Systems \& Control Letters}, vol. 175, p. 105505, 2023.

\bibitem{li2021joint}
L.~Li and M.~Yang, ``Joint localization based on split covariance intersection on the lie group,'' \emph{IEEE Transactions on Robotics}, vol.~37, no.~5, pp. 1508--1524, 2021.

\bibitem{zhang2024towards}
C.~Zhang, J.~Qin, C.~Yan, Y.~Shi, Y.~Wang, and M.~Li, ``Towards invariant extended kalman filter-based resilient distributed state estimation for moving robots over mobile sensor networks under deception attacks,'' \emph{Automatica}, vol. 159, p. 111408, 2024.

\bibitem{niehsen2002information}
W.~Niehsen, ``Information fusion based on fast covariance intersection filtering,'' in \emph{Proceedings of the Fifth International Conference on Information Fusion. FUSION 2002.(IEEE Cat. No. 02EX5997)}, vol.~2.\hskip 1em plus 0.5em minus 0.4em\relax IEEE, 2002, pp. 901--904.

\bibitem{9667208}
Y.~Song, Z.~Zhang, J.~Wu, Y.~Wang, L.~Zhao, and S.~Huang, ``A right invariant extended kalman filter for object based slam,'' \emph{IEEE Robotics and Automation Letters}, vol.~7, no.~2, pp. 1316--1323, 2022.

\bibitem{barfoot2014associating}
T.~D. Barfoot and P.~T. Furgale, ``Associating uncertainty with three-dimensional poses for use in estimation problems,'' \emph{IEEE Transactions on Robotics}, vol.~30, no.~3, pp. 679--693, 2014.

\bibitem{barrau2015non}
A.~Barrau, ``Non-linear state error based extended kalman filters with applications to navigation,'' Ph.D. dissertation, Mines Paristech, 2015.

\bibitem{hartley2020contact}
R.~Hartley, M.~Ghaffari, R.~M. Eustice, and J.~W. Grizzle, ``Contact-aided invariant extended kalman filtering for robot state estimation,'' \emph{The International Journal of Robotics Research}, vol.~39, no.~4, pp. 402--430, 2020.

\bibitem{li2022closed}
X.~Li, H.~Jiang, X.~Chen, H.~Kong, and J.~Wu, ``Closed-form error propagation on $ se\_ $\{$n$\}$(3) $ group for invariant ekf with applications to vins,'' \emph{IEEE Robotics and Automation Letters}, vol.~7, no.~4, pp. 10\,705--10\,712, 2022.

\bibitem{kailath2000linear}
T.~Kailath, A.~H. Sayed, and B.~Hassibi, \emph{Linear estimation}.\hskip 1em plus 0.5em minus 0.4em\relax Prentice Hall, 2000.

\end{thebibliography}

\end{document}